\documentclass[10pt,conference]{IEEEtran}
 
\UseRawInputEncoding
\usepackage{cite}
\usepackage{amsmath}
\usepackage{amsfonts}
\usepackage{graphicx}
\usepackage{float}
\usepackage[table]{xcolor}
\usepackage{subfigure}
\usepackage{epsfig}
\usepackage{amsthm}
\usepackage{mathrsfs}
\usepackage{amsfonts}
\usepackage{amssymb}
\usepackage{mathrsfs}
\usepackage{euscript}
\usepackage{multirow}
\usepackage{url}
\usepackage{algorithm}
\usepackage{algorithmic}
\usepackage{bm}
\usepackage{cuted}
\usepackage{tcolorbox}
\usepackage[font=footnotesize]{subcaption}
\usepackage[font=footnotesize]{caption}
\usepackage{flushend}

\usepackage{hyperref}
 
\usepackage{xcolor, soul}
\newtheorem{theorem}{Theorem}

\newtheorem{corollary}{Corollary}

\sethlcolor{green}
\DeclareMathOperator{\sinc}{sinc}

\belowdisplayskip=\abovedisplayskip
 \usepackage{tikz}

\newcommand\submittedtext{
   \footnotesize This work has been submitted to the IEEE for possible publication. Copyright may be transferred
without notice, after which this version may no longer be accessible.}

 \newcommand\submittednotice{
 \begin{tikzpicture}[remember picture,overlay]
 \node[anchor=north,yshift=-8pt] at (current page.north) {\fbox{\parbox{\dimexpr0.99\textwidth-\fboxsep-\fboxrule\relax}{\submittedtext}}};
 \end{tikzpicture}
 }

\renewcommand\fbox{\fcolorbox{red}{white}}
\theoremstyle{remark}
\newtheorem*{remark*}{Remark}

\DeclareMathOperator{\E}{\mathbb{E}}                     

\begin{document}
\bstctlcite{IEEEexample:BSTcontrol}
\title{Beam Squinting Effects in Super Wideband Communication Systems
\vspace{-3mm}
}

\author{Sachitha C. Bandara$^*$, Peter J. Smith$^\dagger$, Erfan Khordad$^*$, Pawel Dmochowski$^{\dagger\dagger}$, Robin Evans$^*$,  Rajitha Senanayake$^*$\\
$^*$ Department of Electrical and Electronic Engineering, University of Melbourne, Melbourne, Australia\\
$^\dagger$ School of Mathematics and Statistics, Victoria University of Wellington, Wellington, New Zealand\\
$^{\dagger\dagger}$ School of Engineering and Computer Science, Victoria University of Wellington, Wellington, New Zealand
\vspace{-5mm} }

\maketitle
\submittednotice
\begin{abstract}
Beam squint, the frequency-dependent shift of the main beam, poses a major challenge for wideband antenna arrays. 
This paper focuses on the beam squint effects in super wideband (SW) systems, where high mutual coupling (MC) effects are present.
These high MC effects complicate beamforming (BF) by creating frequency-dependent phase relationships that invalidate conventional approaches.
To accurately model MC effects, this paper uses a circuit-theoretic framework for tightly coupled SW uniform linear arrays (ULAs). 
We derive closed-form expressions for the average received signal-to-noise ratio (SNR) with BF in conventional half-wavelength spaced, weakly coupled arrays and validate them. 
Extending our analysis to tightly coupled SW arrays, we demonstrate that, in contrast to conventional weakly coupled arrays, the effective true time delays exhibit a nonlinear dependence on frequency due to coupling-induced phase shifts.
A comparative analysis reveals that strong MC in SW arrays significantly reduces squint in phase-controlled BF, extending the usable bandwidth considerably. 
\end{abstract}

 
 \vspace{-2mm}
\section{Introduction}\label{Sec:Introduction}
Next-generation wireless systems demand higher spatial degrees of freedom and wider bandwidths to meet escalating data requirements\cite{Huo2023}. This drives antenna array designs toward more compact configurations that pack elements in increasingly smaller volumes\cite{mimo_emil}. However, reducing inter-element spacing inevitably increases mutual coupling (MC) between array elements. Traditionally, MC has been viewed as a detrimental effect that degrades array performance~\cite{mutual_coupling_2018,warnick2009}. Recent studies challenge this perspective by demonstrating that intentionally introducing strong MC can yield significant benefits. Tightly coupled super wideband (SW) arrays exploit high MC to achieve massive bandwidth expansions~\cite{super_wideband,connected_arrays,munk2006connected}, substantial endfire gains~\cite{Bandara2025}, reduced channel correlations~\cite{Bandara2025,sachitha2024wcnc}, and new opportunities for multiuser communications~\cite{Bandara2025}. 

High MC in tightly coupled SW arrays introduces complexities for signal processing.
Specifically, the large bandwidth expansion of these arrays makes them particularly susceptible to \textit{beam squint}.
Conventional phase-controlled (PC) beamforming (BF) uses phase shifters calibrated at a single frequency and applied across the entire band with a single radio-frequency (RF) chain \cite{alekseev2017phase}. Since these phase shifts do not adjust to the frequency, the intended beam direction is maintained only at the calibration frequency, causing the beam to steer off-target at other frequencies~\cite{cai2016}. This mismatch, known as beam squint, degrades system performance~\cite{garakoui2011phase}. Time delay (TD) beamformers correct squint by providing frequency-dependent delays \cite{longbrake2012ttd, koch_ttd}. However, TD implementations incur higher hardware complexity~\cite{cai2016,garakoui2011phase}. 
Thus, PC BF remains the preferred solution in many practical systems, offering lower hardware complexity and cost despite its inherent squint limitation.

In tightly coupled arrays, analyzing the effects of MC on beam squinting becomes challenging due to the complex electromagnetic relationships induced by MC. This necessitates a unified framework for communications and electromagnetics, such as circuit theory \cite{Ivrlac_main, damico2023holographic}. 
The existing work on MC effects on beam squinting mainly focuses on suppressing MC \cite{chou2023mc_reduction, abdalrazak2024mc_reduction}, rather than exploiting its potential benefits. Furthermore, TD BF in tightly coupled systems presents additional complexity, as MC between antenna elements can alter effective TDs and introduce residual phase errors that compromise squint cancellation. 
To the best of our knowledge, no prior work has comprehensively analyzed the interplay between MC and beam squint in SW arrays. Furthermore, the practical feasibility of TD solutions for mitigating squint in these tightly coupled systems remains unexplored.
To this end, the insights we draw from our study are critical for next-generation wireless systems, where SW arrays are emerging as one of the key enablers for millimeter-wave and terahertz communications.
The main contributions of this paper are as follows:
\begin{itemize}
    \item We analyze beam squinting effects using a rigorous circuit-theory framework that accurately models MC interactions in SW systems.
    \item We derive circuit-theoretic closed-form expressions for the average received signal-to-noise ratio (SNR) for conventional, weakly coupled half-wavelength spaced uniform linear arrays (ULAs) under PC BF. 
    \item We demonstrate that SW systems cannot achieve complete squint elimination with a single set of time delay units due to frequency-dependent coupling relationships.
    \item We show that tightly coupled SW arrays with phase-controlled BF experience significantly reduced beam squint compared to weakly coupled arrays under similar conditions.
\end{itemize}

\section{System Model}\label{Sec:systemModel}
Let us consider a line-of-sight (LoS) environment where a ULA with $N$ antennas receives signals from a single antenna transmitter.
As our focus is on SW systems with tightly coupled elements, we adopt the circuit-theoretic framework in \cite{Ivrlac_main}, which provides a physically-consistent representation of MC effects. Using input-output voltage relationships, we write the resultant system equation as,
\begin{align}
    \bm{v}_{L}(f) &= \bm{h}(f){v}_{G}(f) + \bm{n}(f),
    \label{eqn:system_eqn}
\end{align}
where $\bm{v}_{L}(f) \in \mathbb{C}^{N\times 1}$ and ${v}_{G}(f)\in \mathbb{C}$ are the frequency domain representations of load (output) and source or generator (input) voltage vectors, respectively. The vector, $\bm{h}(f) \in \mathbb{C}^{N\times 1}$, denotes the equivalent Single-Input-Multiple-Output (SIMO) channel, and $\bm{n}(f) \in \mathbb{C}^{N\times 1}$ represents the noise voltage vector at the receiver. We use frequency-domain representations throughout this paper as they provide insights into how the system behaves at different frequencies across the wide bandwidth. In this setup, the SIMO channel vector can be modelled using the self and mutual impedances of the antenna arrays at both ends and the parameters of other linearly connected devices, such as low noise amplifiers (LNAs) \cite{super_wideband}. For an end-to-end SIMO system with an LNA structure as given in \cite[Fig.~1]{Bandara2025}, and based on  \cite{super_wideband, Ivrlac_main, Bandara2025, sachitha2024wcnc}
we can write the equation for $\bm{h}(f)$ as\footnote{Here, we do not present the full circuit theoretic modelling due to space limitations. For a rigorous, detailed circuit theoretic framework and channel modelling, refer to \cite{Bandara2025, sachitha2024wcnc}.},
\begin{subequations}
\begin{align}
    \bm{h}(f) &= \gamma (f) \mathbf{P}(f) \bm{a}(f),
    \end{align}
    \noindent where,
    \begin{align}
    \gamma (f) &= \rho Z_{\text{LNA}}\sqrt{\Re{\{Z_{T}(f)\}}\Re{\{Z_{R}(f)\}}\beta(f)}{Q}(f)e^{j\psi},\\
    \mathbf{P}(f) &= \left(\mathbf{Z}_{R}(f) + Z_{\text{LNA}}\mathbf{I}_{N}\right)^{-1}.
    \end{align}
    \label{eqn:matrix_relationships_simo} 
\end{subequations}
\par\noindent In (\ref{eqn:matrix_relationships_simo}), $\mathbf{P}(f) \in \mathbb{C}^{N \times N}$ is the receiver coupling matrix, $\mathbf{I}_N$ is an $N \times N$ identity matrix, and $ {Q}(f) = 1/(Z_T(f)+Z_G)$. Note that the above equation is written for far-field (FF) communications. ${Z}_{T}(f)$ is the impedance of the transmitting antenna and $\mathbf{Z}_{R}(f) \in \mathbb{C}^{N\times N}$ is the receiver array impedance matrix, whose diagonal elements are equal to $Z_R(f)$ which represents the self impedance of an antenna.
The off-diagonal elements of $\mathbf{Z}_{R}(f)$ denote mutual impedances between elements.
$Z_G$ denotes impedance of the source, $\rho$ and $Z_{\text{LNA}}$ represent the gain and the impedance of an LNA, \(\psi\) represents a frequency-dependent circuit-theoretic parameter specific to the equivalent RLC circuit of an antenna element, which is given in \cite{super_wideband}, 
$\bm{a}(f)$ represents the steering vector of the LoS channel, which is given by, 
\begin{align}\label{Eq:SteeringVec1}
    \bm{a}(f) &= [1,e^{j2\pi f \frac{\delta}{c}\sin(\phi) },...,e^{j2\pi f \frac{\delta}{c} (N-1)\sin(\phi) }]^T,
\end{align}
with $\phi$ denoting the angle of arrival (AoA) at the receiver ULA 
with respect to the broadside, $\delta$ denoting the inter-element seperation, and $c$ denoting the speed of light. $\beta(f)$ is the channel gain given by, $\beta(f) =  G_{T}G_{R}\left(c/(2\pi fd^{\eta/2})\right)^2$,
where $G_T$ and $G_R$ represent the transmitter and receiver antenna gains, respectively, $d$ denotes distance between the transmitter and the receiver, and $\eta$ represents the path loss exponent.

Calculation of the diagonal elements of $\mathbf{Z}_{R}(f)$ is straightforward and can be done using the equivalent electrical circuits of the corresponding antennas used in the array. For the off-diagonal elements, we adopt a closed-form formula for mutual impedance between two small antennas under the \textit{Chu limit}  which was derived in \cite{shyianov2022, super_wideband}.

Next, we model the noise using the circuit-theoretic framework, which accounts for physically consistent noise sources \cite{Ivrlac_main}, rather than assuming white Gaussian noise. The received noise vector in (\ref{eqn:system_eqn}) for FF communications is given by~\cite{super_wideband},
\begin{align}
    \bm{n}(f) = \bm{v}_{\text{N,LNA}}(f) + \beta Z_{\text{LNA}}\mathbf{P}(f){\bm{v}}_{N,R}(f),
    \label{eqn:noise_equation}
\end{align}
where $\bm{v}_{\text{N,LNA}}(f)\!\in\!\mathbb{C}^{N\times1}$ denotes the LNA noise voltage vector and $\bm{v}_{N,R}(f)\!\in\!\mathbb{C}^{N\times1}$ is the receiver antenna thermal noise.
The corresponding noise covariance matrix is \cite{sachitha2024wcnc},
\begin{align}
    \mathbf{R}_{n}(f)\! &=\! \E\left[\bm{n}(f)\bm{n}^H\!(f)\right]\!\!=\! 4k_{b}T\Delta f \! \left[\Re{\{{Z}_{\text{LNA}}\!\}}(N_{f}\!-\!1)\mathbf{I}_{N} \right. \nonumber \\ 
    & \qquad \left.+ \quad\rho^2{Z}_{\text{LNA}}^2\mathbf{P}(f)\Re\{\mathbf{Z}_R(f)\}\mathbf{P}^{H}(f)\right],
    \label{eqn:noise_corr}
\end{align}
where $k_b$ is the Boltzmann constant, $T$ is the absolute room temperature, $\Delta f$ is the bandwidth, and $\Re{\{{Z}_{\text{LNA}}\}}(N_{f} - 1)$ is the equivalent thermal noise resistance of the LNA.  
$\bm{x}^H$ is the hermitian transpose of $\bm{x}$. 
For simplicity, noise correlations among LNAs due to coupling \cite{warnick2009} are ignored. Next, we will look at the effects of MC in wideband BF.

\section{Analysis of the Beam Squinting Effects}\label{sec:beam_squinting_analysis}

Let us first consider the traditional PC analog BF at the receiver. Based on (\ref{Eq:SteeringVec1}), the conventional analog BF vector, which is designed at the centre frequency, $f_c$, is given by,
\begin{align}
    \bm{w}_{\text{CONV}}(f_c) &= \bm{a}(f_c).
    \label{expr:w_conv}
\end{align}
Here, and in all following sections, we align the beamformer with the signal's AoA. We give an SNR expression that includes the impacts of both coupling and beam squint.
Based on (\ref{eqn:matrix_relationships_simo}) and (\ref{expr:w_conv}), the received SNR of PC analog BF at any given frequency, $f$, can be expressed as,
\begin{align}
    \text{SNR}_{\text{CONV}}(f) &= \frac{|\gamma(f)|^2 |\bm{a}^H(f_c)\mathbf{P}(f)\bm{a}(f)|^2 P_T}{\bm{a}^H(f_c)\mathbf{R}_n(f)\bm{a}(f_c)},
    \label{eqn:conventional_analog_bf}
\end{align}
where $P_T$ is the transmit signal power and we consider equal power allocation for each transmission frequency. In \eqref{eqn:conventional_analog_bf}, we need to distinguish between the frequency of the transmit signal $f$ and the fixed frequency  $f_c$ for which the analog beamformer is designed. The mismatch between $f$ and $f_c$ leads to beam squint.

\subsection{Beam squinting effects of conventional weakly coupled arrays}\label{subsec:squinting_in_conventionak_sys}
Let us first consider (\ref{eqn:conventional_analog_bf}) for a conventional ULA that is designed with the well-known half-wavelength spacing. These arrays are designed to mitigate MC effects. Hence, the coupling and noise correlation effects are negligible, and we can model them as weakly coupled ULAs with coupling and noise correlation matrices given as,
\begin{align}\label{Eq:Sigma2Csigma2nForULA}
    \mathbf{P}(f) &\approx \sigma^2_c(f)\mathbf{I}_N, \nonumber \\
    \mathbf{R}_n(f) &\approx \sigma^2_n(f)\mathbf{I}_N,
\end{align}
respectively, where $\sigma^2_c(f)$ denotes the value of a diagonal element in $\mathbf{P}(f)$ and $\sigma^2_n(f)$ denotes the value of a diagonal element in $\mathbf{R}_n(f)$. Unlike a tightly coupled SW array, a conventional weakly coupled ULA does not provide a massive operational bandwidth; thus, it is used for a relatively smaller bandwidth. As such, we assume that the terms $\gamma(f)$ in \eqref{eqn:conventional_analog_bf} as well as $\sigma^2_c(f)$ and $\sigma^2_n(f)$ in \eqref{Eq:Sigma2Csigma2nForULA} remain approximately constant over the bandwidth.
Hence, for weakly coupled arrays, the expression in (\ref{eqn:conventional_analog_bf}) can be re-expressed as, 
\begin{align}
    \text{SNR}_{\text{CONV}}^{\text{WC}}(f) &= \frac{|\gamma|^2 \sigma^2_c|\bm{a}^H(f_c)\bm{a}(f)|^2 P_T}{\sigma^2_n\bm{a}^H(f_c)\bm{a}(f_c)},
    \label{eqn:conventional_analog_bf_coupling_free}
\end{align}
where the superscript WC denotes a weakly coupled array. 
This can be further simplified to obtain the closed-form expression given by,
\begin{align}
    \text{SNR}_{\text{CONV}}^{\text{WC}}(f) 
    &= \frac{|\gamma|^2 \sigma^2_c P_T }{N \sigma^2_n} 
       \left|\sum_{i=1}^{N} e^{j2\pi (i-1) \frac{\delta}{c}\sin({\phi})(f-f_c)}\right|^2 \nonumber \\
    &= \frac{|\gamma|^2 \sigma^2_c P_T }{N \sigma^2_n} 
       \frac{\sin^2\!\left(\tfrac{N\theta (f)}{2}\right)}{\sin^2\!\left(\tfrac{\theta (f)}{2}\right)},
    \label{eqn:oupling_free_instantaneous}
\end{align}
where $\theta (f) = 2\pi\frac{\delta}{c}(f-f_c)\sin(\phi)$.
Note that this expression includes the beam squinting effect due to the difference between $f$ and $f_c$.

The beam squinting effect can be eliminated by employing true time delay (TTD) beamformers, where each element is equipped with a TD unit. For a weakly coupled ULA, this is achieved by setting the TTD beamforming vector as,
\begin{align}
    \bm{w}_{\text{TTD}}(f) &= \left[1,e^{j2\pi f\Delta t},...,e^{j2\pi f (N-1)\Delta t}\right]^T,
    \label{eqn:TTD_beamformer_geo}
\end{align}
where $\Delta t = \frac{\delta}{c} \sin(\phi)$ is the geometrical time delay, which is uniquely defined by the geometrical properties of the array. Note that the main difference between the TTD beamformer in \eqref{eqn:TTD_beamformer_geo} and the PC beamformer in \eqref{expr:w_conv} is the fixed frequency $f_c$ of the PC beamformer. 
The received instantaneous SNR using a TTD beamformer can be written as,
\begin{align}
    \text{SNR}_{\text{TD}}^{\text{WC}}(f) &=  \frac{|\gamma|^2 \sigma^2_c P_T N}{ \sigma^2_n},
    \label{eqn:oupling_free_ttd}
\end{align}
which eliminates the beam squint effect, because the TTD beamformer adapts to the frequency of the transmit signal.

Next, we consider the average SNR over a bandwidth $\Delta f$, which can be computed using,
\begin{align}
    \overline{\text{SNR}}(\Delta f) &=  \frac{1}{\Delta f}\int_{\Delta f} \text{SNR}(f) df.
    \label{eqn:avg_snr_def}
\end{align}
Here, we do not specify any subscripts or superscripts, as \eqref{eqn:avg_snr_def} is applicable to any BF strategy and any type of array. 
Later, we will use the average SNR to quantify the loss due to beam squinting.
Using this definition, the average received SNR of a weakly coupled receiver array under PC analog BF is stated in the following theorem.

\begin{theorem}
For a weakly coupled ULA with $N$ elements, the average received SNR over a bandwidth $\Delta f$  with PC analog BF is given by,
\begin{align}
\overline{\text{SNR}}_{\text{CONV}}^{\text{WC}}(\Delta f)
&= \frac{|\gamma|^2 \sigma^2_c P_T }{\sigma^2_n} 
\left[
1 + 2\sum_{m=1}^{N-1}\left(1-\frac{m}{N}\right) \right. \nonumber \\
&\qquad\left. \times \sinc\left(\pi\frac{\delta}{c} m\Delta f \sin(\phi)\right)
\right].
\label{eqn:thm1}
\end{align}
\end{theorem}
\noindent\textit{Proof:} See Appendix~A.
\begin{corollary}\label{Coroll:One}
 For a small bandwidth, $ \epsilon$, the SNR expression in \eqref{eqn:thm1} can be written as,
 \begin{align}
\overline{\text{SNR}}_{\text{CONV}}^{\text{WC}}(\epsilon)
&\approx \frac{N|\gamma|^2 \sigma^2_c P_T }{\sigma^2_n} 
\Bigg[
1 - \frac{1}{36} \left(\frac{\pi \delta \epsilon \sin(\phi)}{c}\right)^2  \nonumber
\\ &\qquad
\times (N-1)(N+1)
\Bigg].
\label{cor:1}
\end{align}
\end{corollary}
\noindent To derive (\ref{cor:1}), use the first-order Taylor expansion of the $\sinc(\cdot)$ function in \eqref{eqn:thm1}. 
The resulting expression reveals a quadratic SNR roll-off with bandwidth and provides an accurate approximation for small bandwidths, as validated in Section \ref{Sec:Simulations}.
From (\ref{cor:1}), it can be shown that when the bandwidth reaches zero, the received SNR of a PC beamformer approaches to that of a TTD beamformer as given in (\ref{eqn:oupling_free_ttd}). 

\subsection{Beam squinting effects of tightly coupled SW arrays}\label{ssec:pc_sw}

Next, we discuss the beam squint effect for tightly coupled arrays with non-negligible MC. We divide the discussion into two parts, considering the two BF implementations.
\subsubsection{Phase-controlled BF}
Compared to conventional ULAs where the elements are spaced to avoid MC, SW arrays are made by intentionally increasing MC such that the ULA is tightly coupled. As a result, the receiver steering vector of these arrays is distorted via coupling{\cite{sachitha2024wcnc,Bandara2025}, and the received noise vector in (\ref{eqn:noise_equation}) is correlated due to coupling. 
Hence, the conventional PC analog BF, which relies on the idealized steering vector in (\ref{Eq:SteeringVec1}), is inadequate for SW arrays\cite{Bandara2025}.

For tightly coupled arrays, an alternative PC BF implementation is Phase-Only-Processing (POP) BF. POP BF uses the phase information of the optimal digital matched filter processing, which maximizes the SNR in the analog domain with a single RF chain\cite{Bandara2025}.
The POP BF vector is designed for the center frequency, $f_c$, and is given as \cite{Bandara2025}, 
\begin{align}
    \bm{w}_{\text{POP}}^H(f_c) &= \exp{[j\angle{\gamma^* (f_c) \bm{a} ^H (f_c) \mathbf{P}^{H}(f_c) \mathbf{R}_n^{-1}(f_c)}]},
\end{align}
where $\exp{[j\angle{\bm{x}}]}$ is defined for a generic vector $\bm{x} = [x_1,...,x_N]^T$ as,
\begin{align*}
    \exp{[j\angle{\bm{x}}]} &= \left[e^{j\angle{x_1}}, e^{j\angle{x_1}},...,e^{j\angle{x_N}}\right]^T.  
\end{align*}
Consequently, the instantaneous received SNR for a transmit signal with frequency $f$, using POP BF, is given by
\begin{align}
    \text{SNR}_{\text{POP}}^\text{TC}(f) =  \frac{\left|\bm{w}_{\text{POP}}^H(f_c) \gamma (f) \mathbf{P}(f)\bm{a} (f)\right|^2 P_T}{\bm{w}_{\text{POP}}^H(f_c) \mathbf{R}_n (f)\bm{w}_{\text{POP}}(f_c)},
     \label{eqn:snr_pop}
\end{align}
where the superscript TC represents a tightly coupled array. Note that the expression in (\ref{eqn:snr_pop}) includes the beam squinting effect, as explained for \eqref{eqn:conventional_analog_bf}. For a tightly coupled array, both $\mathbf{P}(f)$ and $\mathbf{R}_n(f)$ exhibit significant non-diagonal entries, reflecting strong element-to-element interactions. Specifically, $\mathbf{P}(f)$ corresponds to the inverse of an impedance matrix, where the non-diagonal elements are defined in \cite[Eqn.~5]{super_wideband}, while $\mathbf{R}_n(f)$ is obtained as in~(\ref{eqn:noise_corr}). Owing to these coupling-induced complexities, it is intractable to derive a simplified formula for the instantaneous received SNR or to obtain a closed-form expression for the average SNR, as done for weakly coupled arrays. 
Thus, we evaluate these metrics using the vector expression in (\ref{eqn:snr_pop}). 

\subsubsection{Time delay BF}\label{ssec:tc_bf_sw}

As we discussed in Section \ref{subsec:squinting_in_conventionak_sys}, for conventional weakly coupled arrays with ideal steering vectors, TD BF with geometry-based frequency-independent delays fully eliminates beam squint. 
However, MC in tightly coupled arrays fundamentally transforms this into a frequency-dependent problem. 
The MC induced distortions cause the phase response of the steering vector to become a nonlinear function of frequency, requiring different TDs at each frequency to compensate for beam squint perfectly.
Consequently, any TD beamformer that perfectly eliminates beam squint in tightly coupled arrays must inherently use frequency-dependent delays.
One such TD beamformer can be defined based on the phase information of the optimal digital matched filter.
Let us consider the received signal vector in (\ref{eqn:system_eqn}). 
\begin{align}
    \bm{v}_L(f) &=  \widetilde{\bm{a}}(f){v}_G(f) + \bm{n}(f),
    \label{eqn:received_vec_sw}
\end{align}
where $\widetilde{\bm{a}}(f) = \gamma(f) \mathbf{P}(f) \bm{a}(f)$ is the distorted steering vector. The optimal beamformer that maximizes SNR while eliminating beam squint is,
\begin{align}
\bm{w}_{\text{opt}}(f) &= \mathbf{R}_n^{-1}(f)\widetilde{\bm{a}}(f).
\end{align}
We can extract the frequency-dependent TDs from the phase response of the optimal beamformer weights.
Let us consider a generic TD BF vector as, 
\begin{align}
\bm{w}_\text{TD}(f) &= \left[e^{j2\pi f \Delta t_1}, e^{j2\pi f \Delta t_2}, \ldots, e^{j2\pi f \Delta t_N}\right]^T,
\label{eqn:ttd_beamformerWeights}
\end{align}
where $\Delta t_k$ is the time delay of $k^{\text{th}}$ element. Hence, to fully eliminate beam squint, $\Delta t_k(f)$, can be set as,
\begin{align}
\Delta t_k(f) &= \frac{\angle{[\mathbf{R}_n^{-1}(f)\widetilde{\bm{a}}(f)}]_k}{2\pi f},
\label{eqn:td_opt}
\end{align}
where $[\bm{x}]_k$ denotes the $k^{\text{th}}$ element of $\bm{x}$.
The frequency dependence of these TDs for tightly coupled arrays complicates the design of a TD beamformer to fully eliminate beam squint, as each array element must be assigned a distinct delay that varies non-linearly with frequency.
This makes it infeasible to realize a single set of frequency-independent TDs that perfectly eliminate squinting. 
A single set of delays that is reasonable over the whole band can be found via numerical optimization methods. 
However, this is left for further work, and we focus on closed-form approaches as described below.

TD-I). \textit{Geometrical TD}: Here we set the geometry-based TD for each element, similarly to a TTD beamformer for weakly coupled arrays. Thus $\Delta t_k$ is given by,
\begin{align}
    \Delta t_k &=\frac{\delta}{c} (k-1) \sin(\phi).
\end{align}

TD-II). \textit{TD at the center frequency}: Instead of implementing TDs at each frequency as given in (\ref{eqn:td_opt}), which complicates designing the TD beamformers, we use a set of TDs at the center frequency for processing the signal at any given frequency. As such, $\Delta t_k$ can be set as,
\begin{align}
    \Delta t_k &=\frac{\angle{[\mathbf{R}_n^{-1}(f_c)\widetilde{\bm{a}}(f_c)}]_k}{2\pi f_c}.
\end{align}
In Section \ref{Sec:Simulations}, we show the effectiveness of our approaches TD-I and TD-II. To obtain squint-free performance as a baseline, we consider TDs as in (\ref{eqn:td_opt}). This provides a way of evaluating TD-I, TD-II, and the loss due to beam squinting in the numerical results.

\section{Numerical Results}\label{Sec:Simulations}
In this section, we validate the derived closed-form expressions for the average received SNR of conventional ULAs, show the effects of high MC on PC and TD BF, and compare the loss due to beam squinting of different BF strategies.

\subsection{Simulation setup}
For conventional weakly coupled ULAs with half-wavelength spacing, the parameters $\gamma$, $\sigma_c^2$, $\sigma_n^2$, and $P_T$ only scale the received SNR and are therefore set to unity.
The array operates at a center frequency of $f_c$=10 GHz. 

For tightly coupled SW arrays, we consider co-linear ULAs composed of canonical minimum scattering (CMS) antennas that satisfy the Chu limit and enable wideband operation \cite{chu1948,kahn1965cms, super_wideband}. 
Unlike the weakly coupled case, coupling and noise correlations over the wide bandwidth are no longer negligible and must be included in the circuit-theoretic model.
The source impedance $Z_G$ and LNA impedance $Z_{\text{LNA}}$ are set to $1\Omega$, while the LNA gain $\rho$ and noise factor $N_f$ are $10$ and $5$ dB, respectively~\cite{super_wideband}.
We compute the mutual impedances between the antenna elements using the closed-form expression in \cite{super_wideband} for small CMS antennas. The inter-element separation $\delta$ is fixed at $0.5$ cm. 
The coupling is controlled by the factor given by $\delta/a_R$, where $a_R$ is the radius of the sphere enclosing an antenna in the receiver array. 
Unless otherwise specified, tight coupling denotes an array with the optimal coupling factor as defined in \cite{super_wideband}, which provides the maximum bandwidth gains.
The radius of the transmitting antenna is set to $100a_R$ to avoid bandwidth limits\cite{super_wideband}. 
The distance between the transmitter and the receiver is fixed at $90$m for all cases, and the path loss exponent, $\eta$, is set to 3.5. Also, the antenna gains, $G_T$, $G_R$ are both fixed at $1.5$\cite{super_wideband}. For all scenarios, we consider an array with $32$ elements. Unless otherwise specified, both AoA and BF angles are set as $\pi/3$ with respect to broadside.

\subsection{Beam squinting effects of conventional ULAs}
\begin{figure}[t] 
    \centering
    \begin{minipage}{0.45\textwidth} 
        \centering
        \includegraphics[width=\textwidth]{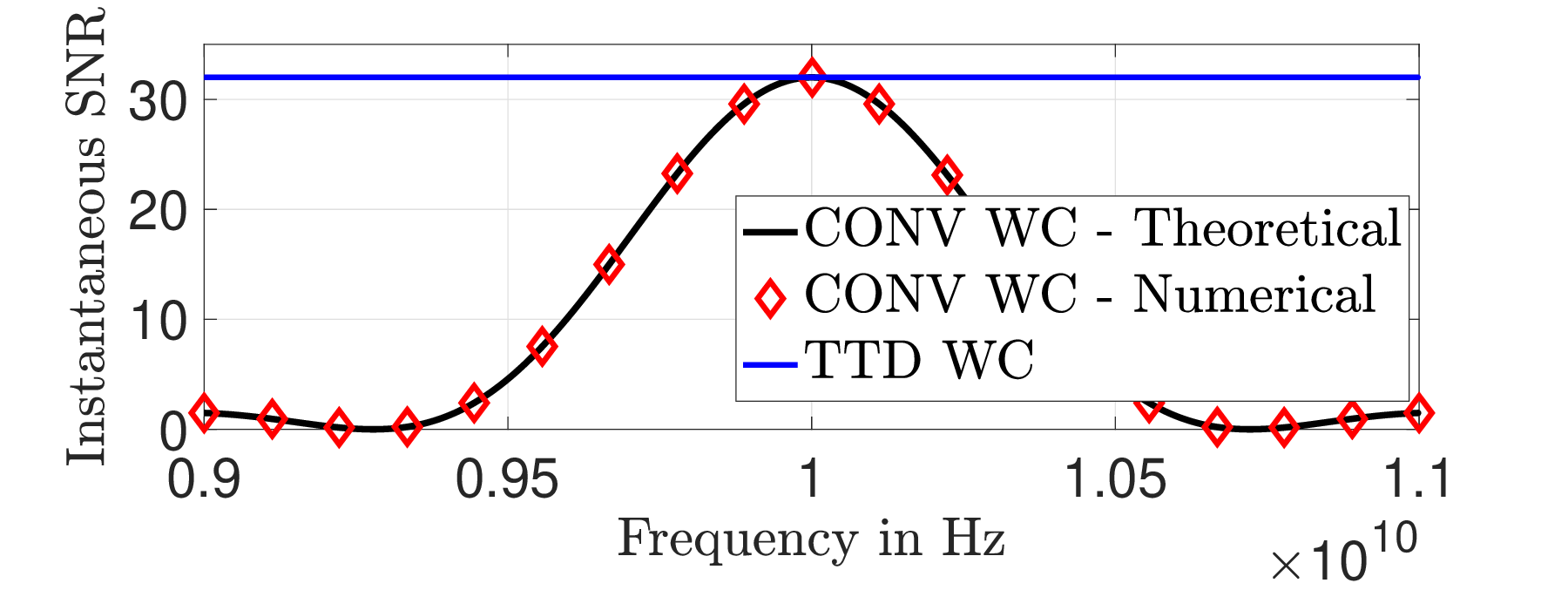}
        \caption*{(a) Instantaneous received SNR vs frequency of weakly coupled arrays}
        \label{fig:instant_valid}
    \end{minipage}
    
    \begin{minipage}{0.45\textwidth} 
        \centering
        \includegraphics[width=\textwidth]{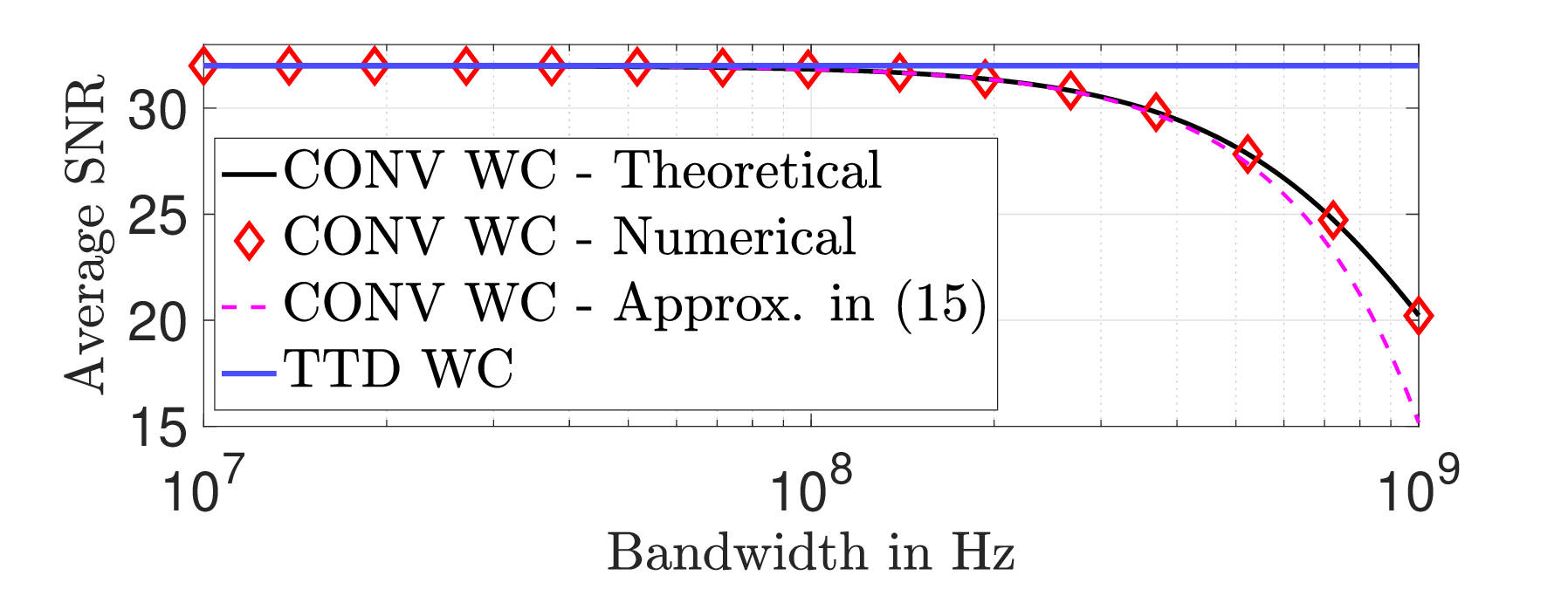}
        \caption*{(b) Average received SNR vs bandwidth of weakly coupled arrays}
        \label{fig:avg_valid}  
    \end{minipage}
    
    \caption{ (a) Instantaneous, and (b) Average received SNR of weakly coupled arrays for $N = 32$, and centered at $10$ GHz.}
    \label{fig:wc_snr_with_squint}
\end{figure}
Fig.~\ref{fig:wc_snr_with_squint} shows the instantaneous and average received SNRs for conventional weakly coupled arrays.
In Fig.~\ref{fig:wc_snr_with_squint}(a), our theoretical expression in (\ref{eqn:oupling_free_instantaneous}) is validated against numerical results for a $2$~GHz bandwidth. 
For conventional analog BF, the SNR decreases when the operating frequency deviates from the center frequency, showing beam squint. 
In contrast, TTD BF keeps the SNR constant across the band, eliminating the beam squint of conventional weakly coupled arrays. 
Fig.~\ref{fig:wc_snr_with_squint}(b) compares the theoretical average SNR in (\ref{eqn:thm1}) with numerical averages for different bandwidths. 
As bandwidth increases, a clear decline in the average received SNR is observed due to beam squint. 
The small-bandwidth approximation in (\ref{cor:1}) is also plotted.
This simplified approximation closely matches the theoretical result for bandwidths up to $ \approx2\%$ of $f_c$, with deviations occurring at wider bandwidths.
\subsection{Beam squinting effects of tightly coupled, SW ULAs}
\begin{figure}[t]
    \centering
    \includegraphics[width=0.44\textwidth]{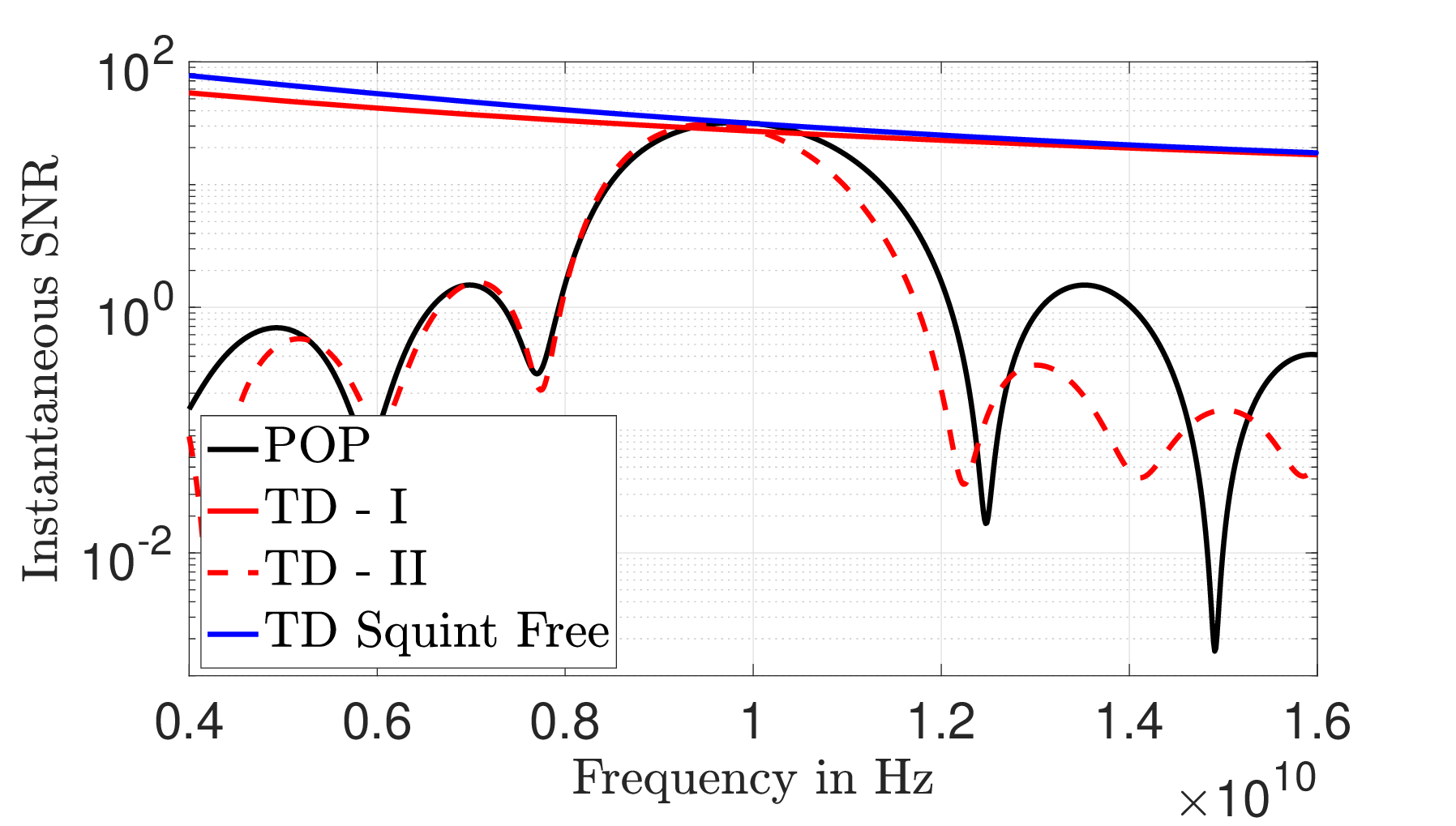}
    \caption{Instantaneous received SNR for tightly coupled SW arrays centered at $10$ GHz, spanning a bandwidth of $12$ GHz.}
    \label{fig:instant_sw}
\end{figure}
Fig.~\ref{fig:instant_sw} shows the instantaneous received SNR for tightly coupled SW arrays centered at $10~$GHz.
It includes results for POP BF as defined in (\ref{eqn:snr_pop}), representing the PC BF approach for SW arrays. The received SNRs for the two TD BF techniques described in Section~\ref{ssec:tc_bf_sw} are also shown. 
The theoretical squint-free TD BF SNR is also included, where TDs are calculated from (\ref{eqn:td_opt}).
All considered BF techniques suffer SNR loss from beam squint, as shown by the lower SNR levels relative to the SNR provided by squint-free TD BF. 
Unlike conventional weakly coupled arrays, both TD BF strategies do not fully suppress beam squinting effects in tightly coupled arrays. 
Notably, the geometrical TD beamformer TD-I significantly outperforms the alternative TD-II, exhibiting near squint-free performance except when extremely wide bandwidths are considered. 
This occurs because TD-II incorporates the exact time delay at the center frequency, accounting for the phase responses of $\mathbf{R_n}(f_c)$, $\mathbf{P}(f_c)$, and $\gamma(f_c)$.
While this ensures exact phase alignment at $f_c$, it introduces an inherent phase misalignment with the channel at other frequencies. In contrast,  the geometric time delay is phase independent and remains useful over a wider frequency range.
Consequently, TD-II achieves superior performance near to the center frequency, whereas TD-I yields better overall wideband performance.

\subsection{Normalized SNR loss due to beam squinting}
\begin{figure}[t] 
    \centering
    \includegraphics[width=0.45\textwidth]{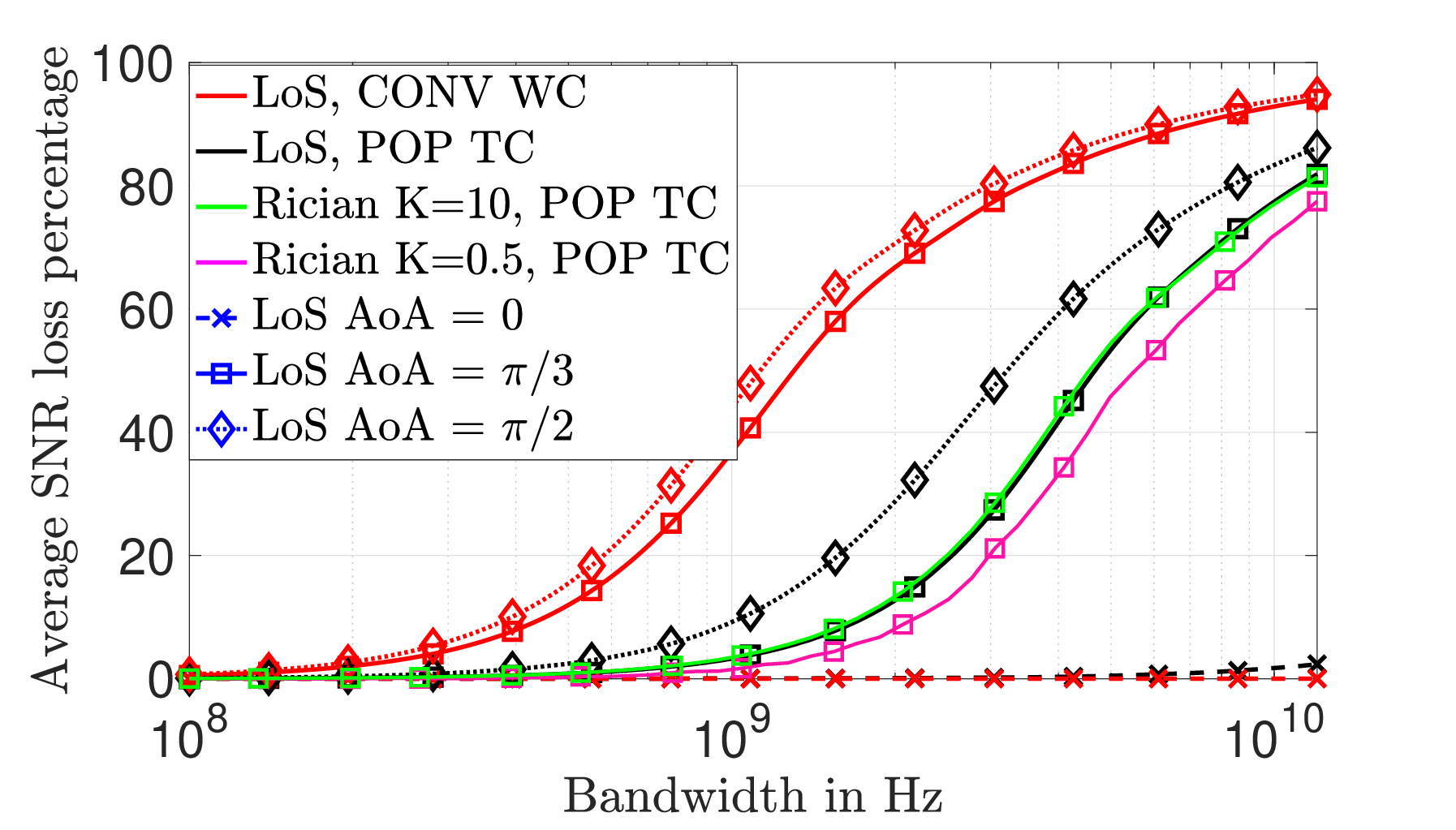}
    \caption{Average SNR loss vs bandwidth for PC BF in conventional weakly coupled arrays and tightly coupled SW arrays, centered at 10~GHz. }
    \label{fig:loss_phased_controlled}
\end{figure}
Fig.~\ref{fig:loss_phased_controlled} shows the normalized SNR loss due to beam squinting, which can be computed using,
\begin{align}
    L_\text{squint} &= \frac{\overline{\text{SNR}}_{\text{TD}}^y (\Delta f) - \overline{\text{SNR}}_x^y (\Delta f)}{\overline{\text{SNR}}_{\text{TD}}^y (\Delta f)} \times 100\%,
\end{align}
where $x$ denotes the BF scheme (CONV, POP) and $y$ the type of ULA (WC, TC). 
The loss is computed relative to TD BF that eliminates squint.
We compare POP BF in a tightly coupled array with strong MC, with PC analog BF in a conventional weakly coupled array under pure LoS.
The loss is plotted versus bandwidths from $100$ MHz to $12$ GHz for various AoAs. 
Strong MC in tightly coupled SW arrays substantially mitigates beam squint compared to conventional weakly coupled arrays. For example, at an AoA of \(\pi/3\) with respect to the broadside, PC BF in weakly coupled arrays incurs a 50\% loss in average SNR due to squint at a bandwidth of \(\approx1.3\) GHz, whereas tightly coupled SW arrays reach the same loss when the bandwidth is \(\approx4.7\) GHz. 
This reduction in beam squint can be attributed to the stronger LoS channel similarity over the band induced by high MC. A detailed investigation of this channel similarity is not included here due to space constraints.
The maximum squint loss occurs at endfire for both array types, but SW arrays exhibit far lower loss at endfire than weakly coupled arrays at \(\phi=\pi/3\). 
This finding is significant since tightly coupled SW arrays provide large SNR gains at endfire\cite{Bandara2025} while still maintaining low squinting losses under phase-controlled POP BF. 
Under phase-controlled POP BF, tightly coupled SW arrays show a small squint loss at broadside when the bandwidth is very large
due to frequency-dependent phase mismatches in the steering vector of the tightly coupled array, unlike weakly coupled arrays that are squint-free at broadside. 
For completeness, we also include two additional curves in Fig.~\ref{fig:loss_phased_controlled} for a multipath scenario with tightly coupled arrays. Rician fading channels from \cite{sachitha2024wcnc} are used with K-factors of 0.5 and 10, representing a substantially scattered and a LoS-dominant environment, respectively. The SNR is computed using the phase of a statistical beamformer \cite{goldsmith} designed at the center frequency, while the squint-free reference SNR is computed in the same way at each subcarrier frequency. As observed, when the LoS component is dominant $(K=10)$, the squinting loss closely resembles that of the pure LoS case. However, as the scattered components become more significant $(K=0.5)$, the phase-controlled statistical beamformer incurs reduced squinting loss compared to the pure LoS environment. This is due to the fact that squinting is less sensitive to multipath channels as the beamformer is targeting the LoS rather than the scattered component.

\section{Conclusion}\label{Sec:conclusion}
In this work, we investigated beam squinting effects of tightly coupled SW arrays using a circuit-theoretic framework that captures MC effects. 
Closed-form expressions for the average received SNR under PC BF in conventional, half-wavelength spaced ULAs were derived and validated. 
We showed that TD BF cannot completely eliminate squint in tightly coupled arrays with a single set of delay units, due to frequency-dependent coupling. 
Numerical analysis demonstrated that strong MC in tightly coupled arrays reduces the loss due to beam squinting for PC BF.
These results indicate that PC beamformers can support wider bandwidths in SW systems, thereby reducing channel knowledge requirements and simplifying implementation.
Future work will explore numerical optimization of TD BF.

\bibliographystyle{IEEEtran}
\bibliography{refference}

@IEEEtranBSTCTL{IEEEexample:BSTcontrol,
   CTLuse_forced_etal       = "yes",
   CTLmax_names_forced_etal = "1",
   CTLnames_show_etal       = "1"
 }

@ARTICLE{super_wideband,
  author={Akrout, Mohamed and Shyianov, Volodymyr and Bellili, Faouzi and Mezghani, Amine and Heath, Robert W.},
  journal={IEEE J. Sel. Areas Commun.}, 
  title={{Super-Wideband Massive MIMO}}, 
  year={2023},
  volume={41},
  number={8},
  pages={2414-2430},
  doi={10.1109/JSAC.2023.3288269}}

@ARTICLE{Ivrlac_main,
  author={Ivrlac, Michel T. and Nossek, Josef A.},
  journal={IEEE Trans. Circuits Syst. I, Reg. Pap.}, 
  title={{Toward a Circuit Theory of Communication}}, 
  year={2010},
  volume={57},
  number={7},
  pages={1663-1683},
  doi={10.1109/TCSI.2010.2043994}}

@ARTICLE{damico2023holographic,
  author={D’Amico, Antonio Alberto and Sanguinetti, Luca},
  journal={IEEE Trans. Wireless Commun.}, 
  title={{Holographic MIMO Communications: What Is the Benefit of Closely Spaced Antennas?}}, 
  year={2024},
  volume={},
  number={},
  pages={},
  doi={10.1109/TWC.2024.3405199}}

@ARTICLE{warnick2009,
  author={Warnick, Karl F. and Woestenburg, Bert and Belostotski, Leonid and Russer, Peter},
  journal={IEEE Trans. Antennas Propag.}, 
  title={{Minimizing the Noise Penalty Due to Mutual Coupling for a Receiving Array}}, 
  year={2009},
  volume={57},
  number={6},
  pages={1634-1644},
  doi={10.1109/TAP.2009.2019898}}

@ARTICLE{kahn1965cms,
  author={Kahn, W.K. and Kurss, H.},
  journal={IEEE Trans. Antennas Propag.}, 
  title={{Minimum-Scattering Antennas}}, 
  year={1965},
  volume={13},
  number={5},
  pages={671-675},
  doi={10.1109/TAP.1965.1138529}}

@ARTICLE{shyianov2022,
  author={Shyianov, Volodymyr and Akrout, Mohamed and Bellili, Faouzi and Mezghani, Amine and Heath, Robert W.},
  journal={IEEE Trans. Commun.}, 
  title={{Achievable Rate With Antenna Size Constraint: Shannon Meets Chu and Bode}}, 
  year={2022},
  volume={70},
  number={3},
  pages={2010-2024},
  doi={10.1109/TCOMM.2021.3099842}}

@article{chu1948,
  title={{Physical Limitations of Omni-directional Antennas}},
  author={Chu, Lan Jen},
  journal={J. Appl. Phys.},
  volume={19},
  number={12},
  pages={1163--1175},
  year={1948},
  publisher={American Institute of Physics}
}

@ARTICLE{mutual_coupling_2018,
  author={Chen, Xiaoming and Zhang, Shuai and Li, Qinlong},
  journal={IEEE Access}, 
  title={{A Review of Mutual Coupling in MIMO Systems}}, 
  year={2018},
  volume={6},
  number={},
  pages={24706-24719},
  doi={10.1109/ACCESS.2018.2830653}}

@phdthesis{connected_arrays,
title = {{Connected Array Antennas : Analysis and Design}},
author = "D. Cavallo",
year = "2011",
doi = "10.6100/IR719461",
language = "English",
isbn = "978-94-6191-035-6",
publisher = "Technische Universiteit Eindhoven",
type = "Phd Thesis 1 (Research TU/e / Graduation TU/e)",
school = "Electrical Engineering",
}

@INPROCEEDINGS{munk2006connected,
       author = {{Munk}, B.~A.},
        title = {{A Wide Band Low Profile Array of End Loaded Dipoles with Dielectric Slab Compensation}},
    booktitle = {EuCAP 2006},
         year = 2006,
       editor = {{Lacoste}, H. and {Ouwehand}, L.},
       series = {ESA Special Publication},
       volume = {626},
        month = oct,
          eid = {9},
        pages = {9},
       adsurl = {https://ui.adsabs.harvard.edu/abs/2006ESASP.626E...9M}
}

@article{Bandara2025,
  author       = {Bandara, Sachitha C. and Smith, Peter J. and Khordad, Erfan and Evans, Robin and Senanayake, Rajitha},
  title        = {{Super Wideband {MIMO} Systems: Channel Modelling and Performance Analysis}},
  journal      = {TechRxiv},
  year         = {2025},
  month        = sep,
  doi          = {10.36227/techrxiv.175832490.02619360/v1},
  note         = {Preprint},
}

@INPROCEEDINGS{sachitha2024wcnc,
  author={Bandara, Sachitha C. and Smith, Peter J. and Khordad, Erfan and Evans, Robin and Senanayake, Rajitha},
  booktitle={2025 IEEE Wireless Commun. Netw. Conf. (WCNC)}, 
  title={{Rician Channel Modelling for Super Wideband MIMO Communications}}, 
  year={2025},
  pages={},
  doi={10.1109/WCNC61545.2025.10978826}}

@article{Huo2023,
  author    = {Yue Huo and Xiaoli Lin and Biao Di and Haijun Zhang and Francisco J. L. Hernando and Alan S. Tan and Shahid Mumtaz and {\"O}zg{\"u}r T. Demir and Kai Chen-Hu},
  title     = {{Technology Trends for Massive {MIMO} towards 6{G}}},
  journal   = {Sensors},
  year      = {2023},
  volume    = {23},
  number    = {13},
  pages     = {6062},
  doi       = {10.3390/s23136062},
  pmid      = {37447911},
  pmcid     = {PMC10347082},
  month     = {June}
}

@ARTICLE{mimo_emil,
  author={Björnson, Emil and Eldar, Yonina C. and Larsson, Erik G. and Lozano, Angel and Poor, H. Vincent},
  journal={IEEE Signal Process. Mag.}, 
  title={{Twenty-Five Years of Signal Processing Advances for Multiantenna Communications: From Theory to Mainstream Technology}}, 
  year={2023},
  volume={40},
  number={4},
  pages={107-117},
  doi={10.1109/MSP.2023.3261505}}

@INPROCEEDINGS{cai2016,
  author={Cai, Mingming and Gao, Kang and Nie, Ding and Hochwald, Bertrand and Laneman, J. Nicholas and Huang, Huang and Liu, Kunpeng},
  booktitle={2016 IEEE Glob. Commun. Conf.}, 
  title={{Effect of Wideband Beam Squint on Codebook Design in Phased-Array Wireless Systems}}, 
  year={2016},
  volume={},
  number={},
  pages={1-6},
  doi={10.1109/GLOCOM.2016.7841766}}

@INPROCEEDINGS{longbrake2012ttd,
  author={Longbrake, Matt},
  booktitle={2012 IEEE Natl. Aerosp. Electron. Conf. (NAECON)}, 
  title={{True Time-Delay Beamsteering for Radar}}, 
  year={2012},
  volume={},
  number={},
  pages={246-249},
  doi={10.1109/NAECON.2012.6531062}}

@article{alekseev2017phase,
  author = {V.I. Alekseev and A.N. Anufriev and V.P. Meshchanov and others},
  title = {{New Structure of Ultrawideband Fixed Phase Shifters Based on Stepped Coupled Transmission Lines with Stubs}},
  journal = {J. Commun. Technol. Electron.},
  volume = {62},
  number = {},
  pages = {535--541},
  year = {2017},
  doi = {10.1134/S1064226917050011},
}

@INPROCEEDINGS{garakoui2011phase,
  author={Garakoui, Seyed Kasra and Klumperink, Eric A. M. and Nauta, Bram and van Vliet, Frank E.},
  booktitle={2011 41st Eur. Microw. Conf. (EuMC)}, 
  title={{Phased-array Antenna Beam Squinting Related to Frequency Dependency of Delay Circuits}}, 
  year={2011},
  volume={},
  number={},
  pages={1304-1307},
  doi={10.23919/EuMC.2011.6101846}}

@article{koch_ttd,
title={{Timed Array Architectures and Integrated True-Time Delay Elements for Wideband Millimeter-Wave Antenna Arrays}}, 
volume={17}, 
DOI={10.1017/S1759078724001053}, 
number={2}, 
journal={Int. J. Microw. Wirel. Technol.}, 
author={Koch, Manuel and Schönhärl, Stefan and Breun, Sascha and Fischer, Georg and Weigel, Robert}, 
year={2025}, 
pages={190–201}}

@ARTICLE{chou2023mc_reduction,
  author={Chou, Hsi-Tseng and Chang, Chen-Yi and Torrungrueng, Danai},
  journal={IEEE Access}, 
  title={{Minimization of Mutual Coupling Interferences Between Nearby Antenna Arrays to Retain Their Beam Steering Functionality}}, 
  year={2023},
  volume={11},
  number={},
  pages={87430-87441},
  doi={10.1109/ACCESS.2023.3305253}}

@INPROCEEDINGS{abdalrazak2024mc_reduction,
  author={Abdalrazak, Mariam Q. and Majeed, Asmaa H. and Abd-Alhameed, Raed A. and Khaleel, Sherif A. and Salama, Ahmed Gamal and Mabrouk, Mohamed and Metwally, Ibrahim M and Amar, Ahmed S. I.},
  booktitle={2024 Int. Telecommun. Conf. (ITC-Egypt)}, 
  title={{Investigating the Impact of Mutual Coupling on mmWave UPA with Consideration of Beam Squint}}, 
  year={2024},
  volume={},
  number={},
  pages={667-671},
  doi={10.1109/ITC-Egypt61547.2024.10620519}}

@ARTICLE{goldsmith,
  author={Goldsmith, A. and Jafar, S.A. and Jindal, N. and Vishwanath, S.},
  journal={IEEE J. Sel. Areas Commun.}, 
  title={{Capacity Limits of MIMO Channels}}, 
  year={2003},
  volume={21},
  number={5},
  pages={684-702},
  doi={10.1109/JSAC.2003.810294}}

\begin{appendices}
\section{Proof of Theorem 1}
Using (\ref{eqn:avg_snr_def}) and (\ref{eqn:oupling_free_instantaneous}), the average received SNR for a weakly coupled ULA can be written as,
\begin{align}
    & \overline{\text{SNR}}_{\text{CONV}}^{\text{WC}}(\Delta f) \nonumber \\
&\quad = \frac{1}{\Delta f} \frac{|\gamma|^2 \sigma^2_c P_T }{N \sigma^2_n} \int_{\Delta f} \left|\sum_{i=1}^{N} e^{j2\pi (i-1) \frac{\delta}{c}\sin({\phi})(f-f_c)}\right|^2 df \nonumber \\
& \quad= \frac{1}{\Delta f} \frac{|\gamma|^2 \sigma^2_c P_T }{N \sigma^2_n} \!\sum_{i=1}^{N}\sum_{e=1}^{N} \! \int_{f_c- \frac{\Delta f}{2}}^{f_c+ \frac{\Delta f}{2}}\!\! e^{j2\pi \frac{\delta}{c}(i-e)\sin({\phi})(f-f_c)}df \nonumber \\
&\quad = \frac{|\gamma|^2 \sigma^2_c P_T }{N \sigma^2_n} \!\sum_{i=1}^{N}\sum_{e=1}^{N} \sinc\left( \pi \frac{\delta}{c}\Delta f (i-e) \sin(\phi)\right).
\label{pf:stp1}
\end{align}
Let $\alpha = \pi \frac{\delta}{c}\Delta f \sin(\phi)$. Then (\ref{pf:stp1}) can be written as,
\begin{align}
    & \overline{\text{SNR}}_{\text{CONV}}^{\text{WC}}(\Delta f)  = \frac{|\gamma|^2 \sigma^2_c P_T }{N \sigma^2_n}  \left[\sum_{i=1}^{N}\sinc(\alpha(i-1)) \right.\nonumber \\
    &\quad + \left.\sum_{i=1}^{N} \sinc(\alpha(i-2))+\cdots+\sum_{i=1}^{N} \sinc(\alpha(i-N))\right]\nonumber \\
    &\quad \!=\! \frac{|\gamma|^2 \sigma^2_c P_T }{N \sigma^2_n} \!\left[\sum_{m=1}^{N-1}\!2(N\!-\!m)\sinc(m\alpha) + N\sinc(0)\!\right]\!,
\end{align}
which yields (\ref{eqn:thm1}).
\end{appendices}

\end{document}